\documentclass[a4paper,fleqn,10pt]{cas-dc}
\usepackage[numbers]{natbib}
\usepackage{amsmath,amssymb,amsthm}
\usepackage{color}
\usepackage{chngcntr}
\usepackage{indentfirst}
\usepackage{enumitem}
\usepackage{subfig}
\usepackage{bbm}
\usepackage{siunitx}
\usepackage{xcolor}
\usepackage[linesnumbered,ruled,vlined]{algorithm2e}

\SetCommentSty{mycommfont}

\SetKwInput{KwInput}{Input}                
\SetKwInput{KwOutput}{Output}              
\SetKw{Continue}{continue}

\usepackage{optidef}
\usepackage{cleveref}

\def\tsc#1{\csdef{#1}{\textsc{\lowercase{#1}}\xspace}}
\tsc{WGM}
\tsc{QE}
\tsc{EP}
\tsc{PMS}
\tsc{BEC}
\tsc{DE}

\renewcommand{\printorcid}{}  

\theoremstyle{case}
\newtheorem{case}{Case}

\newtheorem{remark}{Remark}
\newtheorem{prop}{Proposition}[section]
\newtheorem{theorem}{Theorem}[section]

\newtheorem*{theorem*}{Theorem}
\newtheorem*{prop*}{Proposition}

\DeclareMathOperator*{\argmin}{arg\,min}

\usepackage{tikz}
\usetikzlibrary{matrix,positioning}
\tikzset{bullet/.style={circle,fill,inner sep=2pt}}
\usepackage{adjustbox}
\usepackage{pgfplots}

\begin{document}
\let\WriteBookmarks\relax
\def\floatpagepagefraction{1}
\def\textpagefraction{.001}
\shorttitle{DRNV}
\shortauthors{Singh et~al.}

\title [mode = title]{Distributionally Robust Newsvendor with Moment Constraints}                      



\author[1]{Derek Singh}
\cormark[1]
\ead{singh644@umn.edu}
\cortext[cor1]{Corresponding author}

\author[1]{Shuzhong Zhang}
\ead{zhangs@umn.edu}

\address[1]{Department of Industrial and Systems Engineering, University of Minnesota, Minneapolis, MN 55455}

\ExplSyntaxOn
\keys_set:nn { stm / mktitle } { nologo }
\ExplSyntaxOff

\begin{abstract}
This paper expands the work on distributionally robust newsvendor to incorporate moment constraints. The use of Wasserstein distance as the ambiguity measure is preserved. 
The infinite dimensional primal problem is formulated; problem of moments duality is invoked to derive the simpler finite dimensional dual problem. An important research question is: How does distributional ambiguity affect the optimal order quantity and the corresponding profits/costs? To investigate this, some theory is developed and a case study in auto sales is performed. We conclude with some comments on directions for further research.
\end{abstract}

\begin{keywords}
robust newsvendor \sep moments duality \sep distributionally robust optimization \sep Wasserstein distance \sep Lagrangian duality
\end{keywords}

\maketitle

\section{Introduction and Overview}
\subsection{Newsvendor Model}
The newsvendor model is a classical problem in inventory control that has been extensively studied under a variety of settings. It is a decision problem in which the decision maker wants to determine the optimal order quantity to balance overage and underage costs to minimize losses (or alternatively maximize profits) during the next sales period, between orders. Overage costs result from ordering too much product, in excess of realized demand, which leads to waste. Underage costs result from ordering too little, short of realized demand, which leads to lost sales. The classical setting is a single decision period, for a non-perishable product, where the demand distribution is known. Scarf \citep{scarf1958min} solved the original distributionally robust formulation, which assumes knowledge of both first and second moments. \par
In this paper, we extend that work to allow for more precise information regarding the underlying distribution, as measured in terms of Wasserstein distance. This allows us to modulate the degree of ambiguity and investigate its effect on the optimal order quantity and the corresponding profits/losses.
This is done via the framework of Wasserstein discrepancy between distributions and the problem of moments duality results. First we formulate the distributionally robust newsvendor decision problem to account for ambiguity in distribution. Next, we formulate and solve the simpler finite dimensional dual problem. After that, we use industry data to investigate the practical effects of ambiguity. \par
An outline of this paper is as follows. Section 1 gives on overview of the preliminary material as well as a literature review. Section 2 develops the main theoretical results to characterize the distributionally robust newsvendor model with moment constraints. Section 3 conducts a case study using a five year historical data set of monthly sales of Tesla automobiles. Section 4 discusses conclusions and suggestions for further research. All detailed proofs are deferred to the Appendix. \par

\subsection{Literature}
As mentioned previously, the original work in this problem domain was initaited by Scarf \citep{scarf1958min}. Since then, further research has been done on the topic of distributional robustness. One paper of note by Gallego and Moon \citep{gallego1993distribution} uses a different choice of parameters and a presents an alternate proof for Scarf's ordering rule. According to the authors, this choice of parameters and method of proof are simpler and more interpretable. The development of distributionally robust optimization with moment-based ambiguity sets led to extensions that considered multiple items, risk aversion, and additional distributional properties such as asymmetry and multimodality \citep{hanasusanto2015distributionally,natarajan2018asymmetry}. \par
More recent, and alternative, work in distributional robustness is focused on the idea of distributionaly ambiguity as measured through the statistical distance between distributions. A few of the seminal papers in this area, that use the Wasserstein distance framework, are by Gao and Kleywegt \citep{gao2016distributionally}, Esfahani and Kuhn \citep{Esfahani17}, and Blanchet and Karthyek \citep{blanchetFirst}. Applications of this framework include topics in machine learning, portfolio management, inventory control, and others. A relevant example is the work by Lee et al.\ \citep{lee2020data}, that studies the data-driven distributionally robust risk-averse \\ newsvendor model under Wasserstein ambiguity. The authors make the case for dispensing with the first two moment constraints (difficulty in estimation, challenges in obtaining this information in a practical setting) and analyze the resulting solutions in settings with/without risk aversion and for Wasserstein distances of order $p=1$ and $p > 1$.

\subsection{Preliminary Material}
\subsubsection{Definitions}
The empirical measure $P_n$ is defined as
$P_n := \frac{1}{n} \sum_{i=1}^{n} \delta_{x_i}$
where $x_i$ denotes the $i\textsuperscript{th}$ realization of demand and $\delta_{x_i}$ is a Dirac measure.
The ambiguity set for probability measures is
$U_{\delta}(P_n) = \{P: W_d(P,P_n) \leq \delta\}$
where $W_d$ is the Wasserstein discrepancy with associated distance function $d(\cdot,\cdot)$. 
\begin{equation*}
W_d(P,P') = \inf_\pi \{ \mathbb{E}^\pi[d(X,Y)]: X \sim P, Y \sim P' \}.
\end{equation*}
Here $d(X,Y)$ is the distance between random variables $X$ and $Y$ that follow distributions $P$ and $P'$ respectively, and the $\inf$ is taken over all joint distributions $\pi$ with marginals $P$ and $P'$. To simplify the analysis, this work uses the (squared) Euclidean distance $d(x,y) = \| x-y \|^2_2 = \sum_{i=1}^n (x_i - y_i)^2$ \citep{zhao2018data}. 

\subsubsection{Moments Duality}
In Section 2 we formulate the primal decision problem for the distributionally robust newsvendor with moment constraints. Problem of moments duality results will be used to formulate the simpler yet equivalent dual problem. In this context, to specify the first moment constraint $M_1(X) = \mathbb{E}^P[ X ] = \mu$ and the second moment constraint $M_2(X) = \mathbb{E}^P[ X^2 ] = (\mu^2 + \sigma^2)$, we appeal to the strong duality of linear semi-infinite programs. The dual problem appears to be more tractable than the primal problem since it only involves the (finite dimensional) empirical measure $P_n$ as opposed to a continuum of probability measures. This allows us to solve a data-driven nested optimization problem defined by the chosen data set. A brief restatement of the duality result can be found in \cite{singh2020tight}. See Appendix B of \cite{blanchet2019robust} and Proposition 2 of \cite{blanchetMV} or the original work \cite{isii1962sharpness} for further details. \par

\section{Theory: DRNV}
This section formulates the primal and dual problems for the distributionally robust newsvendor (DRNV). A polynomial time algorithm, using the directional descent (DD) method, is developed to compute the solution to the dual.
\subsection{Equivalent Objective Functions}
This work will use the cost minimization form of the newsvendor decision problem. In this subsection, we show the equivalence between the cost minimization and profit maximization forms. The cost (objective) function is 
\begin{equation*}
\Pi_1(Q) = c_1 \mathbb{E} (X-Q)^+ + c_2 \mathbb{E} (Q-X)^+
\end{equation*}
where $c_1$ and $c_2$ denote the underage and overage costs, X denotes random demand, and Q denotes the order quantity (decision variable).
The profit (objective) function is
\begin{equation*}
\Pi_2(Q) = (p-s) \, \mathbb{E} \min(Q,X) - (c-s) Q
\end{equation*}
where $p > c > s > 0$ for sale price $p$, unit cost $c$, and salvage value $s$. Let us use the relations
\begin{align*}
\mathbb{E} \min(Q,X) &= Q - \mathbb{E} (Q-X)^+ = \mu - \mathbb{E} (X-Q)^+ \\
\mathbb{E} \max(Q,X) &= Q + \mathbb{E} (X-Q)^+ = \mu + \mathbb{E} (Q-X)^+
\end{align*}
to substitute for $Q$, $\min(Q,X)$, in the profit function to get
\begin{align*}
\Pi_2(Q) &= (p-s) (Q - \mathbb{E} (Q-X)^+) - (c-s) Q \\
&= (s-p) \mathbb{E} (Q-X)^+ + (p-c) Q \\
&= (s-p) \mathbb{E} (Q-X)^+ + (p-c) (\mu + \mathbb{E} (Q-X)^+ \nonumber \\
&- \mathbb{E} (X-Q)^+) \\
&= (p-c) (\mu - \mathbb{E} (X - Q)^+) - (c-s) \mathbb{E} (Q-X)^+.
\end{align*}
Observe that maximizing $\Pi_2(Q)$ is equivalent to minimizing $\Pi_1(Q)$ for $c_1 = p-c > 0$ and $c_2 = c-s > 0$.
\subsection{Primal Problem}
Let us characterize constraints via the set 
\small
\[ C := {\{ X \in \mathbb{R}_+: \\ M_1(X) = \mu, M_2(X) = (\mu^2+\sigma^2), P \in U_{\delta}(P_n) \}}. \]
\normalsize
where $\mathbb{R}_+$ denotes the restriction to non-negative realized demand.
The primal problem can then be written as 
\begin{equation*}
\inf_{Q \geq 0} \sup_{X \in C} \; c_1 \mathbb{E} (Q-X)^+ + c_2 \mathbb{E} (X-Q)^+ \tag{P}
\end{equation*}
where $Q$ is the order quantity, $X$ is demand, and $\delta$ is the radius of the Wasserstein ball of distributions centered at $P_n$.
\subsection{Dual Problem}
Following the approach in \citep{singh2020tight}, let $\lambda := (\lambda_1,\lambda_2,\lambda_3)$, define $\Gamma := {\{ \lambda : \lambda_1 \geq 0, \lambda_2, \lambda_3 \}}$, and apply moments duality to write the dual problem as
\begin{equation*} \label{d1}
\inf_{Q \geq 0} \inf_{\lambda \in \Gamma} \; F(\lambda,Q) :=  \lambda_1 \delta + \lambda_2 \mu + \lambda_3 (\sigma^2 + \mu^2) + \\ \sum_{i=1}^n \Psi_i(\lambda,Q) \tag{D}
\end{equation*}
where $\Psi_i(\lambda,Q) := -\lambda_1 x_i^2 + \sup_{x \geq 0} g_i(x,\lambda,Q)$ for
\begin{equation*}
g_i(x,\lambda,Q) := c_1 (x-Q)^+ + c_2(Q-x)^+ - a x^2 + 2b_i x
\end{equation*}
for $a := \lambda_1+\lambda_3$, $b_i := \lambda_1 x_i - \lambda_2/2$. The constraint $a > 0$ ensures $F(\lambda,Q)$ has finite value. Note that (\ref{d1}) is jointly convex in $(\lambda,Q)$. Our approach to solve (\ref{d1}) has two steps: (i) solve the inner problem in closed form and develop a polynomial time algorithm to compute its solution, (ii) use a bisection method to compute the solution to the outer problem. Overall, the algorithm finishes in polynomial time.
\begin{prop}
\setlength{\parindent}{0pt}
Let $a > 0$, $Q \geq 0$, then
\small
\[
\sup_{x \geq 0} g_i(x,\lambda,Q) = \begin{cases}
g_{i0}, & \text{if } (\lambda_1,\lambda_2) \in \tilde{R}_{i0}, \\
g_{i1}, & \text{if } (\lambda_1,\lambda_2)  \in  \tilde{R}_{i1}, \\
g_{i2}, & \text{if } (\lambda_1,\lambda_2)  \in  \tilde{R}_{i2}, 
\end{cases}
\]
\normalsize
where $x_{i1} = (2b_i - c_2)/(2a)$, $x_{i2} = (2b_i+c_1)/(2a)$, 
$\tilde{c} = c_1+c_2$, the $g_{ij}$ functions are given by
\begin{align*}
g_{i0} &= c_2 Q, \\
g_{i1} &= c_2(Q-x_{i1}) -a x_{i1}^2 + 2 b_i x_{i1}, \\
g_{i2} &= c_1(x_{i2}-Q) -a x_{i2}^2 + 2 b_i x_{i2}, 
\end{align*}
the mapping of regions to optimal solutions $x^*$ is given by Table 1,
\begin{table}[!htb]
  \begin{center}
    \caption{Mapping of Regions to Optimal Solutions.}
    \label{tab:table1}
    \begin{tabular}{|c|c|} 
	\hline
      \textbf{Region $\tilde{R}_{ij}$} & \textbf{Optimal Solution $x^*$} \\
      \hline
      $\tilde{R}_{i0}$ & $0$ \\
	 \hline	
      $\tilde{R}_{i1}$ & $x_{i1}$ \\
	\hline
      $\tilde{R}_{i2}$ & $x_{i2}$ \\
	\hline
    \end{tabular}
  \end{center}
\end{table}
and the $\tilde{R}_{ij}$ regions are defined for four cases, according to the relative positions of $\lambda_2$ intercepts $\beta_1 := c_1$, $\beta_2 := -c_2$, $\beta_3 := (c_1-c2)/2 - 2aQ$, $\beta_4 := c_1 -2\sqrt{a\tilde{c}Q}$.
\setcounter{case}{0}
\begin{case}
$\beta_3 \geq \beta_2 \cap \beta_4 \in [\beta_2,\beta_1]$. \\
\[
\tilde{R}_{i0} :=  R_{i1} \cup R_{i2}, \tilde{R}_{i1} := \varnothing, \tilde{R}_{i2} := R_{i3} \cup R_{i4},
\]
where regions $R_{ij} \subseteq \{ \lambda_1 \geq 0, \lambda_2 \}$  are given by
\begin{align*}
R_{i1} &:= \{ x_{i2} \leq 0 \}, \\
R_{i2} &:= \{ x_{i2} \geq 0 \cap \lambda_2 - 2x_i\lambda_1 \geq \beta_4 \}, \\
R_{i3} &:= \{ x_{i1} \leq 0 \cap \lambda_2 - 2x_i\lambda_1 \leq \beta_4 \}, \\
R_{i4} &:= \{ x_{i1} \geq 0 \}. 
\end{align*}
\end{case}

\begin{case}
$\beta_3 < \beta_2 \cap \beta_4 \in [\beta_2,\beta_1]$. \\
\[
\tilde{R}_{i0} := R_{i1} \cup R_{i2}, \tilde{R}_{i1} := R_{i4}, \tilde{R}_{i2} := R_{i3} \cup R_{i5},
\]
where $R_{ij}$ regions are given by
\begin{align*}
R_{i1} &:= \{ x_{i2} \leq 0 \}, \\
R_{i2} &:= \{ x_{i2} \geq 0 \cap \lambda_2 - 2x_i\lambda_1 \geq \beta_4 \}, \\
R_{i3} &:= \{ x_{i1} \leq 0 \cap \lambda_2 - 2x_i\lambda_1 \leq \beta_4 \}, \\
R_{i4} &:= \{ x_{i1} \geq 0 \cap \lambda_2 - 2x_i\lambda_1 \geq \beta_3 \}, \\
R_{i5} &:= \{ \lambda_2 - 2x_i\lambda_1 \leq \beta_3 \}. 
\end{align*}
\end{case}

\begin{case}
$\beta_3 < \beta_2 \cap \beta_4 \notin [\beta_2,\beta_1]$. \\
\[
\tilde{R}_{i0} :=  R_{i1}, \tilde{R}_{i1} := R_{i3}, \tilde{R}_{i2} := R_{i2}  \cup R_{i4},
\]
where $R_{ij}$ regions are given by
\begin{align*}
R_{i1} &:= \{ x_{i2} \leq 0 \}, \\
R_{i2} &:= \{ x_{i2} \geq 0 \cap x_{i1} \leq 0 \}, \\
R_{i3} &:= \{ x_{i1} \geq 0 \cap \lambda_2 - 2x_i\lambda_1 \geq \beta_3 \}, \\
R_{i4} &:= \{ \lambda_2 - 2x_i\lambda_1 \leq \beta_3 \}.
\end{align*}
\end{case}

\begin{case}
$\beta_3 \geq \beta_2 \cap \beta_4 \notin [\beta_2,\beta_1]$. 
\[
\tilde{R}_{i0} := R_{i1}, \tilde{R}_{i1} := \varnothing, \tilde{R}_{i2} := R_{i2} \cup R_{i3},
\]
where $R_{ij}$ regions are given by
\begin{align*}
R_{i1} &:= \{ x_{i2} \leq 0 \}, \\
R_{i2} &:= \{ x_{i1} \leq 0 \cap x_{i2} \geq 0 \}, \\
R_{i3} &:= \{ x_{i1} \geq 0 \}. 
\end{align*}
\end{case}
\end{prop}
\begin{proof}
The proof first solves the unconstrained problem to deduce that the set of possible values for optimal solution $x^*$ for the constrained problem is $S := \{0,x_{i1},x_{i2}\}$. From there, some analysis is done, for each of the four cases, to partition the $(\lambda_1 \geq 0, \lambda_2)$ half-plane into regions and to construct the mapping between regions and optimal solutions. Substituting the appropriate $x^*$ per region into $g_i$ gives the result. See Appendix for the detailed proof.
\end{proof}

\begin{theorem}
DD method evaluates $f(\xi,Q) := \min_{\{\lambda_1 \geq 0,\lambda_2\}} \\  F(\lambda_1,\lambda_2,\xi,Q)$ for $\xi := \lambda_1+\lambda_3 > 0$, in polynomial time.
\end{theorem}
\begin{proof}
Note the DD method can evaluate $f(\xi,Q)$ in at most $\mathcal{O}(n^2)$ operations, as it searches $\mathcal{O}(n^2)$ line segments and regions that partition the $\{ \lambda_1 \geq 0, \lambda_2 \}$ half-plane, and it is a descent method that only needs to traverse each line segment and/or region once. This once-only traversal property holds due to the joint convexity of $F(\lambda_1,\lambda_2,\xi,Q)$.
\end{proof}

\begin{algorithm}[!htb]
\renewcommand{\thealgocf}{}
\SetAlgorithmName{DD method}{} 
\DontPrintSemicolon
	\KwInput{$\{\xi\:,Q\:,\: \{x_i\}\:,\: N\:,n\:,\:\delta\:,\:\mu\:,\:\sigma\}$}  
	\KwOutput{$\{y_{\xi,Q} = f(\xi,Q)\}$}
	Sort $\{ x_i \}$ Decreasing \; 
	Define intercepts $\{\beta_j\}$ as in Proposition 2.1\;
	Construct lines $\{\lambda_2 = L_{ij}(\lambda_1 \geq 0) \}$ where $L_{ij}( \lambda_1 ) :=  2\lambda_1 x_i + \beta_j \; \forall j $\;
	Compute $\{ V_\tau \}$, the set of vertices $(\lambda_1,\lambda_2)$ where either two or more lines $L_{ij}$ intersect or $\lambda_2 \in \{ L_{ij}(\lambda_1=0) \} $ \;
    	Set $k=0$ and the initial search point to $\lambda_c(k) = V_{\tau^*}$, the vertex with the smallest value for $F$ \;
	\While {$k < N$} 
	{
		Search adjacent regions $\Gamma$ for descent directions $\lambda^\circ_c(k) + t d_\gamma$ where we move towards the min value $\lambda^*_\gamma$ for $F_\Gamma$ \;
		\tcc*[h]{Here $F_\Gamma$ is defined such that $\{\Psi_i\}$ have the same functional form across the entire $(\lambda_1, \lambda_2)$ plane as in region $\Gamma$, where $\Gamma$ is defined by the supporting lines $L_{ij}. \;\; \lambda^\circ_c(k)$ is an interior point to region $\Gamma$ within $\epsilon$ of  $\lambda_c(k).$ Observe that the number of regions $\Gamma$ is $\mathcal{O}(n^2)$.} 
		\\
		\If{ $F(\lambda^*_\gamma) < F(\lambda_c(k))$ }
		{
			\If{ $\lambda^\circ_c(k) + t d_\gamma \: \cap \: \{ L_i \} = \emptyset$ }
		    {
				$\lambda_c(k+1) := \lambda^*_\gamma$ \;
		    }
		    \Else
		    {
				$\{\lambda_m\} :=  \lambda^\circ_c(k) + t d_\gamma \: \cap \: \{ L_{ij} \}$ \;
				$\lambda_c(k+1) := \argmin_{\{\lambda_m\}} \| \lambda_c(k) - \lambda_m \|$ \;
		    }
			$k = k+1$ \;
			\Continue \;

		}
		Search along adjacent rays $R$ (the line segments $\pm{\vec{L_{ij}}}$ emanating from point $\lambda_c(k)$) for descent directions $\lambda_c(k) + t d_r$ where we move towards a critical point $\lambda^*_r$ with zero directional derivative for $F$, so $D_{d_r} \, F(\lambda^*_r) = 0$ \;
		\If{ $\{d_r : D_{d_r} \, F(\lambda^*_r) = 0 \} \neq \emptyset$ }
		{
			$\lambda_c(k+1) := \argmin_{\{ \lambda^*_r \}} \| \lambda_c(k) - \lambda^*_r \|$  \;
			$k = k+1$ \;
		}
		\Else
		{
			\tcc*[h]{There are no descent directions via regions or rays so we are at the min value.} \\
			\Return $y_{\xi,Q} = F(\lambda_c(k))$ \;
		}
		
	}
\caption{Directional Descent method to compute $f(\xi,Q)$ for (D1)}
\end{algorithm}

\begin{remark}
The equation $\lambda_2 = L_{ij}(\lambda_1)$ is derived from algebraic manipulation and substitution for $(\lambda_1,\lambda_2)$ in $Q = \tilde{Q}$.
\end{remark}

\begin{prop}
The solution to (\ref{d1}) can be computed in polynomial time.
\end{prop}
\begin{proof}
The dual problem (\ref{d1}) is jointly convex in $(\lambda,Q)$ hence $f(\xi,Q)$ is jointly convex. For fixed $(\xi,Q)$, $f(\xi,Q)$ can be evaluated, using the DD method, in at most $\mathcal{O}(n^2)$ operations to find the (global) minimum of a piecewise convex quadratic function in $(\lambda_1,\lambda_2)$. Thus, one can apply a line search method to evaluate $h(Q) := \min_{\xi > 0} f(\xi,Q)$ with $Q$ fixed. The constraint $\xi > 0$ ensures the piecewise quadratics have finite local minima. Finally, one can apply a bisection method to the subgradient of $h$ to evaluate $\min_{Q \geq 0} h(Q)$ to compute the solution to (\ref{d1}). Note that $\max(c_1,c_2)$ is an upper bound on the magnitude of the subgradient of $h$; this establishes Lipschitz continuity of $h$ and ensures the bisection method will complete in polynomial time.
\end{proof}


\section{Case Study}
Let us now investigate a practical application of the theory and algorithms developed in this work, using 4 years of monthly U.S. Tesla auto sales data from \citep{teslaSales}. Sales are recorded in units of thousands and lost sales penalty $c_1 \approx 20$ whereas overage penalty $c_2 \approx 10$ (in units of thousands of dollars). 
We compute the (worst case) expected costs (in millions) and optimal order quantity (in thousands) as a function of ambiguity. Figure 2 shows the trajectories approach the classical (Scarf) limits of approximately 121mm USD in cost and 12,930 in order quantity. Thus, for any alternate demand distribution of Wasserstein distance at most $\delta$, that satisfies the moment constraints, one can look up its worst case cost and order quantity. To go further, one could map the level of ambiguity $\delta$ to a statistical confidence level $\alpha \in [0,1]$ (that the Wasserstein ball contains the true distribution) using measure concentration results \cite{Carlsson2018}. \par
For simplicity of implementation, a grid search in $(\lambda_1, \lambda_2)$, as opposed to the DD method, is used to evaluate $f(\xi,Q)$. A line search is used to evaluate $h(Q) := \min_{\xi > 0} f(\xi,Q)$ with $Q$ fixed, and the bisection method is used to evaluate $\min_{Q \geq 0} h(Q)$.
The algorithms are coded in Matlab and make use of standard functions such as \textit{bisection} and \textit{fminbnd}. No special Matlab toolboxes are needed (although parallel computing via \textit{parfor} loops for the grid search in $(\lambda_1,\lambda_2)$ requires use of that toolbox). 

\begin{figure}[!htb]
	\centering
	\caption{DD Method: Plots for $\xi=1$}%
	\subfloat[$L_{ij}$ Lines]{\scalebox{0.2}[0.2]{\includegraphics{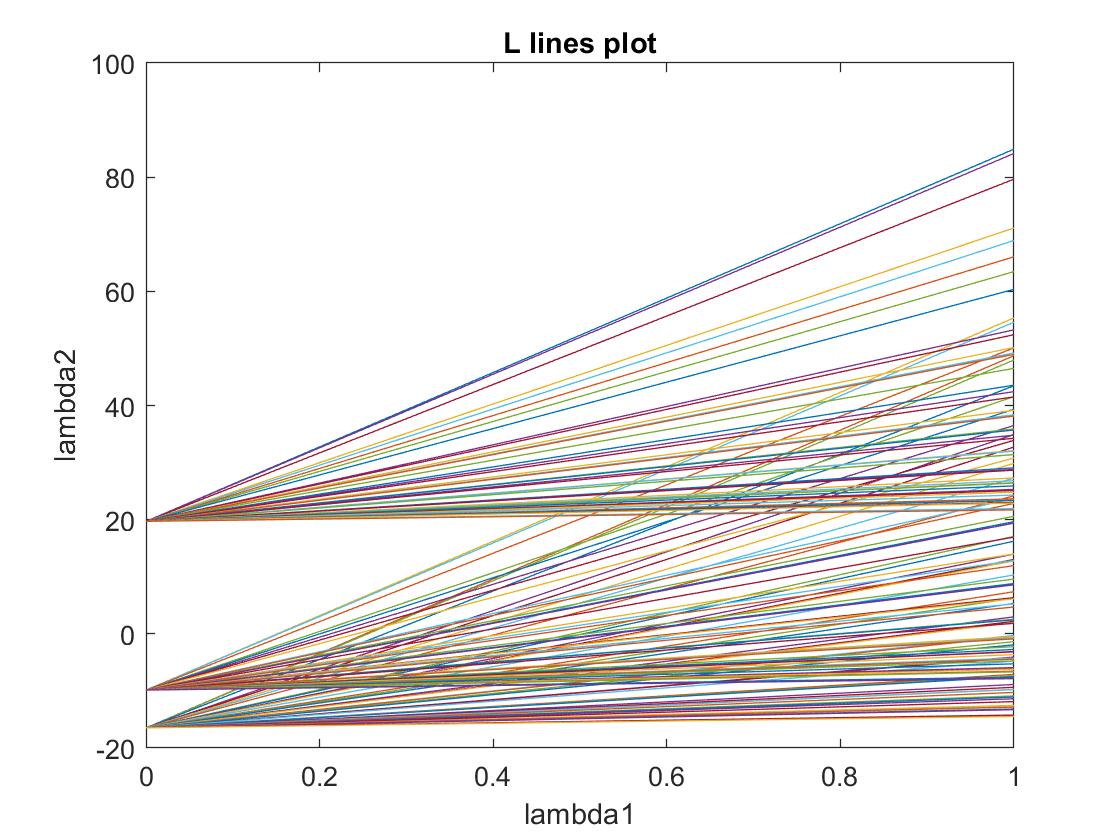}}}%
	\quad
	\subfloat[Surface Plot]{\scalebox{0.225}[0.25]{\includegraphics{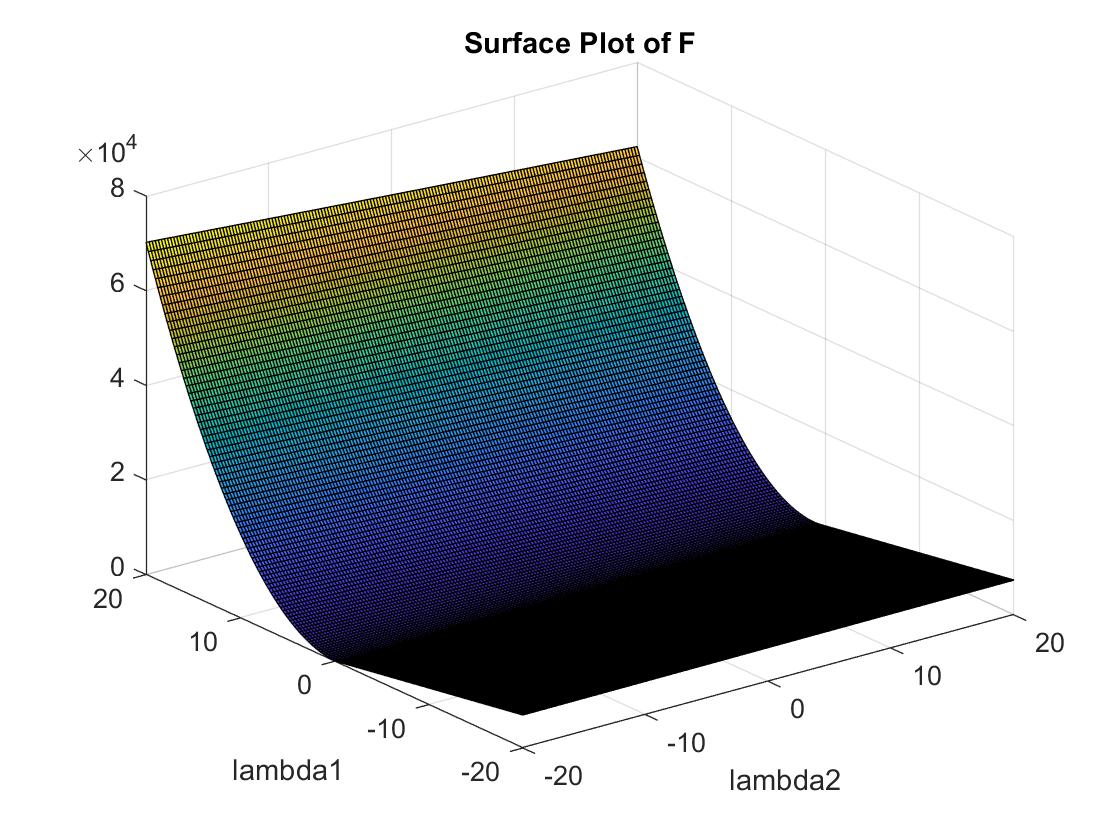}}}%
\end{figure}

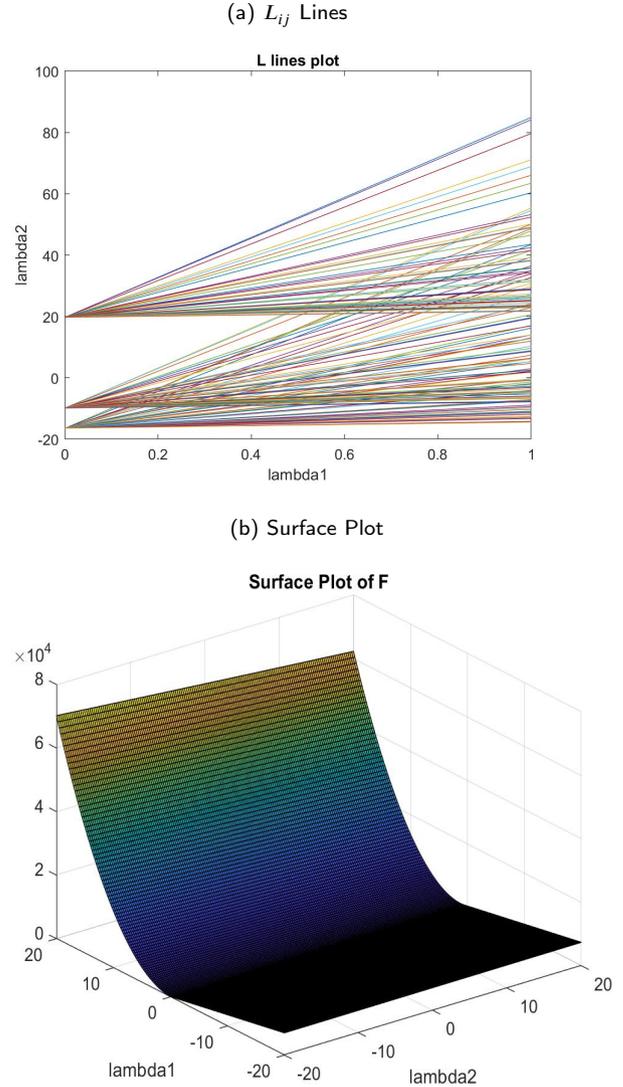
\begin{figure}[!htb]
\caption{Costs and Optimal Order Quantity}
\begin{center}
\scalebox{0.85}{
\begin{tikzpicture}

\pgfplotsset{set layers}

	\begin{axis}[legend pos=north west,
		xlabel=Delta,
		ylabel=Cost (mm),
		ymin = 100, ymax = 140,
		axis y line*=left ]
	\addplot[color=magenta,mark=x] coordinates {
		(0,103)
		(1,110)
		(2.5,113)
		(5,116)
		(10,119)
		(15,120)
		(20,121)
	}; \label{plot1_y1}
	\end{axis}

	\begin{axis}[legend pos=north west,
		xlabel=Delta,
		ylabel=Q (thousands),
		ymin = 10, ymax = 14,
		yticklabel pos=right,
		axis y line*=right,
		ylabel near ticks,
		axis x line = none ]
	\addplot[color=cyan,mark=*] coordinates {
		(0,11)
		(1,11.4)
		(2.5,11.7)
		(5,12)
		(10,12.3)
		(15,12.44)
		(20,12.6)
	}; \label{plot1_y2}


    \addlegendimage{/pgfplots/refstyle=plot1_y1}\addlegendentry{Q}
	\addlegendimage{/pgfplots/refstyle=plot1_y2}\addlegendentry{Cost}
	\end{axis}
\end{tikzpicture}
}
\end{center}
\end{figure}

\section{Conclusions and Further Work}
This work has analyzed the distributionally robust newsvendor model (with moment constraints) using Wasserstein distance as the ambiguity measure. The robust newsvendor decision problem and requisite preliminary material were covered in Section 1. The simpler dual problem was formulated and solved in Section 2. Furthermore, a polynomial time algorithm, using the DD method, was developed to compute the solution. In Section 3, we presented a case study on Tesla automobile sales. Finally, we conclude with some commentary on directions for further research. \par
One direction for future research would be to analyze DRNV in a multiproduct and/or multiperiod setting, using the tools of semidefinite programming (SDP) and multi-stage optimization. Another direction for future research would be to explore risk-averse formulations of DRNV. Finally, perhaps a third direction for future research would be to investigate alternate forms for the objective function. \par

\section*{Data and Code Availability Statement}
The raw and/or processed data and Matlab code required to reproduce the findings from this research can be obtained from the corresponding author, [D.S.], upon reasonable request.

\section*{Conflict of Interest Statement}
The authors declare they have no conflict of interest.

\section*{Funding Statement}
The authors received no specific funding for this work.

\bibliographystyle{cas-model2-names}

\bibliography{DRNV}


\appendix
\section{Proof of Proposition 2.1}
\begingroup
\setlength{\parindent}{0pt}
Let $a > 0$, $Q \geq 0$, then
\small
\[
\sup_{x \geq 0} g_i(x,\lambda,Q) = \begin{cases}
g_{i0}, & \text{if } (\lambda_1,\lambda_2) \in \tilde{R}_{i0}, \\
g_{i1}, & \text{if } (\lambda_1,\lambda_2)  \in  \tilde{R}_{i1}, \\
g_{i2}, & \text{if } (\lambda_1,\lambda_2)  \in  \tilde{R}_{i2}, 
\end{cases}
\]
\normalsize
where $x_{i1} = (2b_i - c_2)/(2a)$, $x_{i2} = (2b_i+c_1)/(2a)$, 
$\tilde{c} = c_1+c_2$, the $g_{ij}$ functions are given by
\begin{align*}
g_{i0} &= c_2 Q, \\
g_{i1} &= c_2(Q-x_{i1}) -a x_{i1}^2 + 2 b_i x_{i1}, \\
g_{i2} &= c_1(x_{i2}-Q) -a x_{i2}^2 + 2 b_i x_{i2}, 
\end{align*}
the mapping of regions to optimal solutions $x^*$ is given by Table 2,
\begin{table}[!htb]
  \begin{center}
    \caption{Mapping of Regions to Optimal Solutions}
    \label{tab:table1}
    \begin{tabular}{|c|c|} 
	\hline
      \textbf{Region $\tilde{R}_{ij}$} & \textbf{Optimal Solution $x^*$} \\
      \hline
      $\tilde{R}_{i0}$ & $0$ \\
	 \hline	
      $\tilde{R}_{i1}$ & $x_{i1}$ \\
	\hline
      $\tilde{R}_{i2}$ & $x_{i2}$ \\
	\hline
    \end{tabular}
  \end{center}
\end{table}
and the $\tilde{R}_{ij}$ regions are defined for four cases, according to the relative positions of $\lambda_2$ intercepts $\beta_1 := c_1$, $\beta_2 := -c_2$, $\beta_3 := (c_1-c2)/2 - 2aQ$, $\beta_4 := c_1 -2\sqrt{a\tilde{c}Q}$.
\setcounter{case}{0}
\begin{case}
$\beta_3 \geq \beta_2 \cap \beta_4 \in [\beta_2,\beta_1]$. \\
\[
\tilde{R}_{i0} :=  R_{i1} \cup R_{i2}, \tilde{R}_{i1} := \varnothing, \tilde{R}_{i2} := R_{i3} \cup R_{i4},
\]
where regions $R_{ij} \subseteq \{ \lambda_1 \geq 0, \lambda_2 \}$  are given by
\begin{align*}
R_{i1} &:= \{ x_{i2} \leq 0 \}, \\
R_{i2} &:= \{ x_{i2} \geq 0 \cap \lambda_2 - 2x_i\lambda_1 \geq \beta_4 \}, \\
R_{i3} &:= \{ x_{i1} \leq 0 \cap \lambda_2 - 2x_i\lambda_1 \leq \beta_4 \}, \\
R_{i4} &:= \{ x_{i1} \geq 0 \}.
\end{align*}
\end{case}

\begin{case}
$\beta_3 < \beta_2 \cap \beta_4 \in [\beta_2,\beta_1]$. \\
\[
\tilde{R}_{i0} := R_{i1} \cup R_{i2}, \tilde{R}_{i1} := R_{i4}, \tilde{R}_{i2} := R_{i3} \cup R_{i5},
\]
where $R_{ij}$ regions are given by
\begin{align*}
R_{i1} &:= \{ x_{i2} \leq 0 \}, \\
R_{i2} &:= \{ x_{i2} \geq 0 \cap \lambda_2 - 2x_i\lambda_1 \geq \beta_4 \}, \\
R_{i3} &:= \{ x_{i1} \leq 0 \cap \lambda_2 - 2x_i\lambda_1 \leq \beta_4 \}, \\
R_{i4} &:= \{ x_{i1} \geq 0 \cap \lambda_2 - 2x_i\lambda_1 \geq \beta_3 \}, \\
R_{i5} &:= \{ \lambda_2 - 2x_i\lambda_1 \leq \beta_3 \}. 
\end{align*}
\end{case}

\begin{case}
$\beta_3 < \beta_2 \cap \beta_4 \notin [\beta_2,\beta_1]$. \\
\[
\tilde{R}_{i0} :=  R_{i1}, \tilde{R}_{i1} := R_{i3}, \tilde{R}_{i2} := R_{i2} \cup R_{i4},
\]
where $R_{ij}$ regions are given by
\begin{align*}
R_{i1} &:= \{ x_{i2} \leq 0 \}, \\
R_{i2} &:= \{ x_{i2} \geq 0 \cap x_{i1} \leq 0 \}, \\
R_{i3} &:= \{ x_{i1} \geq 0 \cap \lambda_2 - 2x_i\lambda_1 \geq \beta_3 \}, \\
R_{i4} &:= \{ \lambda_2 - 2x_i\lambda_1 \leq \beta_3 \}.
\end{align*}
\end{case}

\begin{case}
$\beta_3 \geq \beta_2 \cap \beta_4 \notin [\beta_2,\beta_1]$. 
\[
\tilde{R}_{i0} := R_{i1}, \tilde{R}_{i1} := \varnothing, \tilde{R}_{i2} := R_{i2} \cup R_{i3},
\]
where $R_{ij}$ regions are given by
\begin{align*}
R_{i1} &:= \{ x_{i2} \leq 0 \}, \\
R_{i2} &:= \{ x_{i1} \leq 0 \cap x_{i2} \geq 0 \}, \\
R_{i3} &:= \{ x_{i1} \geq 0 \}. 
\end{align*}
\end{case}

\begin{proof}
It turns out to be useful to solve the \textit{unconstrained} problem $\sup_{x \in \mathbb{R}} g_i$ first.
To start, use the $\max$ relation 
\begin{equation}
(Q-x)^+ = (Q-x) + (x-Q)^+
\end{equation}
to express function $g_i(x,\lambda,Q)$ as
\begin{equation} 
g_i = (c_1+c_2)(x - Q)^+ + c_2(Q-x) -a x^2 + 2b_i x. 
\end{equation}
The first order optimality conditions (for left and right derivatives) for (unconstrained) $\sup_{x \in \mathbb{R}} g_i$ say that 
\begin{align}
(c_1+c_2) \mathbbm{1}_{(0,\infty)}(x - Q) -c_2 - 2ax + 2b_i \geq 0, \\ 
(c_1+c_2) \mathbbm{1}_{[0,\infty)}(x - Q) -c_2 - 2ax + 2b_i \leq 0,
\end{align}
which leads to three cases for the critical point $x^*$. 
\setcounter{case}{0}
\begin{case}
$x^* > Q$\\
$x^* > Q \implies x^* = x_{i2} \; \text{for} \; Q < x_{i2}$.
\end{case}
\begin{case}
$x^* < Q$\\
$x^* < Q \implies x^* = x_{i1} \; \text{for} \; Q > x_{i1}$.
\end{case}
\begin{case}
$x^* = Q$\\
This case violates the first order optimality conditions and hence does not occur.
\end{case}
Consequently, there are three cases for $Q$.
\setcounter{case}{0}
\begin{case}
$Q \leq x_{i1}$\\
$x^* =  x_{i2} \implies g_i = \hat{g}_i = \tilde{c}(x_{i2} - Q) + 2 a x_{i1} x_{i2} - a x_{i2}^2 + c_2 Q$.
\end{case}
\begin{case}
$Q \geq x_{i2}$\\
$x^* = x_{i1} \implies g = \check{g}_i = a x_{i1}^2 + c_2 Q$.
\end{case}
\begin{case}
$Q \in \big(x_{i1},x_{i2}\big)$\\
$x^* \in \{ x_{i1}, x_{i2} \} \implies g = \max{ \hat{g}_i, \check{g}_i } \implies \\ g_i
= \begin{cases}
	\hat{g}_i, & \text{if }  x_{i1} < Q \leq \tilde{Q},\\
	\check{g}_i, & \text{if } \tilde{Q} < Q < x_{i2}.  
	\end{cases}
$
\end{case}
The value $(x_{i1}+x_{i2})/2$ for $\tilde{Q}$ is derived by equating function values for $\hat{g}_i$ and $\check{g}_i$.
Therefore, for the unconstrained problem the $(\lambda_1 \geq 0, \lambda_2)$ half-plane is partitioned into \textit{two} regions above and below
the cut line $Q = \tilde{Q} = (x_{i1}+x_{i2})/2$.
Now let us consider the constrained problem $\sup_{x \geq 0} g_i$. We can reason that the optimal solution $x^*$ is such that $x^* \in S_{012} := \{0,x_{i1},x_{i2}\}$. It remains to determine for which regions of the half-plane that $x^*$ takes on one of the values in $S_{012}$. 
As it turns out, there are four cases depending on the relative positions of the intercepts $\{ \beta_1,\beta_2,\beta_3,\beta_4 \}$.
\begin{table}[!htb]
  \begin{center}
    \caption{Intercepts}
    \label{tab:table1}
    \begin{tabular}{|c|c|} 
	\hline
      \textbf{Intercept $\beta_{j}$} & Interpretation \\
      \hline
      $\beta_{1}$ & $x_{i2} = 0$ \\
	 \hline	
      $\beta_{2}$ & $x_{i1} = 0$ \\
	\hline
      $\beta_{3}$ & $g_{i1} = g_{i2}$ \\
	\hline
      $\beta_{4}$ & $g_{i0} = g_{i2}$ \\
	\hline
    \end{tabular}
  \end{center}
\end{table}

\setcounter{case}{0}
\begin{case}
$\beta_3 \geq \beta_2 \cap \beta_4 \in [\beta_2,\beta_1]$. \\
Refer to Figure 3 and Table 3 for the analysis. Let us work with the following form for the equations for lines $L_{ij}$, with intercepts $\beta_j$, that cut the half-plane into regions:
\begin{equation}
\lambda_2 = 2\lambda_1 x_i + \beta_j.
\end{equation}
add also define sets $S_{12} := \{ x_{i1}, x_{i2} \}$ and $S_{02} := \{ 0, x_{i2} \}$. 
The only choice for $x^* \in S_{012}$ for $R_{i1}$ is 0. Since $\beta_3 > \beta_2$, the choice $x_{i1}$ for $x^*$ is not available above $\beta_2$.
Furthermore, the choice for $x^*$ in $R_{i4}$ is $x_{i2}$.
Some algebra gives that the choice for $x^* \in S_{02}$ is 0 for $R_{i2}$ and $x_{i2}$ for $R_{i3}$. 
Substituting the appropriate $x^*$ per region into $g_i$ gives the result.

\begin{figure}
\centering
\begin{tikzpicture}
    \draw [thin, gray, ->] (0,-2) -- (0,2)      
        node [above, black] {$\lambda_2$};  
    
    \draw [thin, gray, ->] (-1,0) -- (5,0)      
        node [right, black] {$\lambda_1$};    
    
    \draw [draw=brown,] (0,1) -- (5,2);	
    \draw [draw=green,] (0,0) -- (5,1);	
    \draw [draw=blue,] (0,-1) -- (5,0);	

    \node [left] at (0,1) {$\beta_1$};      
    \node [left] at (0,0) {$\beta_4$};      
    \node [left] at (0,-1) {$\beta_2$};      

   \node [above] at (2.5,1.5) {$R_{i1}$};      
   \node [above] at (2.5,0.5) {$R_{i2}$};      
   \node [above] at (2.5,-0.5) {$R_{i3}$};      
   \node [above] at (2.5,-1.5) {$R_{i4}$};      

\end{tikzpicture}
\caption{Case 1: Cut Lines for $(\lambda_1 \geq 0, \lambda_2)$ Half-Plane}
\end{figure}
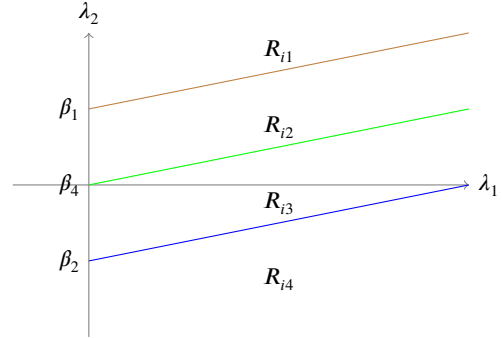
\end{case}

\begin{case}
$\beta_3 < \beta_2 \cap \beta_4 \in [\beta_2,\beta_1]$. \\
Refer to Figure 4 and Table 3 for the analysis. The approach is similar to Case 1. 
The only choice for $x^*$ for $R_{i1}$ is 0. As before, the choice for $x^* \in S_{02}$ for $R_{i2}$ is $0$ and $x_{i2}$ for $R_{i3}$.
Some algebra gives that the choice for $x^* \in S_{12}$ is $x_{i1}$ for $R_{i4}$ and $x_{i2}$ for $R_{i5}$. 
As before, substitution for $x^*$ per region into $g_i$ gives the result.

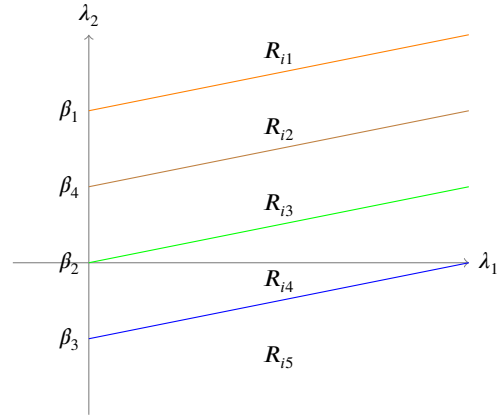
\begin{figure}
\centering
\begin{tikzpicture}
    \draw [thin, gray, ->] (0,-2) -- (0,3)      
        node [above, black] {$\lambda_2$};  
    
    \draw [thin, gray, ->] (-1,0) -- (5,0)      
        node [right, black] {$\lambda_1$};    
    
    \draw [draw=orange,] (0,2) -- (5,3);	
    \draw [draw=brown,] (0,1) -- (5,2);	
    \draw [draw=green,] (0,0) -- (5,1);	
    \draw [draw=blue,] (0,-1) -- (5,0);	

    \node [left] at (0,2) {$\beta_1$};      
    \node [left] at (0,1) {$\beta_4$};      
    \node [left] at (0,0) {$\beta_2$};      
    \node [left] at (0,-1) {$\beta_3$};      

   \node [above] at (2.5,2.5) {$R_{i1}$};      
   \node [above] at (2.5,1.5) {$R_{i2}$};      
   \node [above] at (2.5,0.5) {$R_{i3}$};      
   \node [above] at (2.5,-0.5) {$R_{i4}$};      
   \node [above] at (2.5,-1.5) {$R_{i5}$};      

\end{tikzpicture}
\caption{Case 2: Cut Lines for $(\lambda_1 \geq 0, \lambda_2)$ Half-Plane}
\end{figure}
\end{case}

\begin{case}
$\beta_3 < \beta_2 \cap \beta_4 \notin [\beta_2,\beta_1]$. \\
Refer to Figure 5 and Table 3 for the analysis. 
The only choice for $x^*$ for $R_{i1}$ is 0. Since $\beta_4 < \beta_2$, the choice for $x^* \in S_{02}$ for $R_{i2}$ is $x_{i2}$.
As before, the choice for $x^*$ on either side of $\beta_3$ is $x_{i1}$ for $R_{i3}$ and $x_{i2}$ for $R_{i4}$.

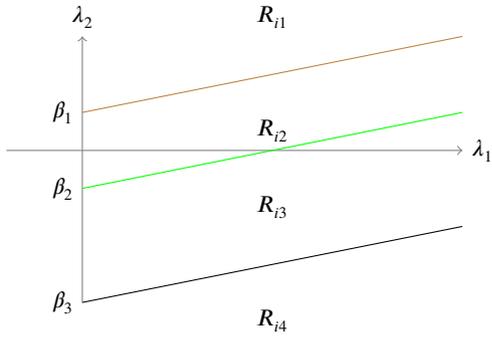
\begin{figure}
\centering
\begin{tikzpicture}
    \draw [thin, gray, ->] (0,-2) -- (0,1.5)      
        node [above, black] {$\lambda_2$};  
    
    \draw [thin, gray, ->] (-1,0) -- (5,0)      
        node [right, black] {$\lambda_1$};    
    
    \draw [draw=brown,] (0,0.5) -- (5,1.5);	
    \draw [draw=green,] (0,-0.5) -- (5,0.5);	
    \draw [draw=black,] (0,-2) -- (5,-1);	

    \node [left] at (0,0.5) {$\beta_1$};      
    \node [left] at (0,-0.5) {$\beta_2$};      
    \node [left] at (0,-2) {$\beta_3$};      

   \node [above] at (2.5,1.5) {$R_{i1}$};      
   \node [above] at (2.5,0) {$R_{i2}$};      
   \node [above] at (2.5,-1) {$R_{i3}$};      
   \node [above] at (2.5,-2.5) {$R_{i4}$};      

\end{tikzpicture}
\caption{Case 3: Cut Lines for $(\lambda_1 \geq 0, \lambda_2)$ Half-Plane}
\end{figure}
\end{case}

\begin{case}
$\beta_3 \geq \beta_2 \cap \beta_4 \notin [\beta_2,\beta_1]$. \\
Refer to Figure 6 and Table 3 for the analysis. 
The only choice for $x^*$ for $R_{i1}$ is 0. Since $\beta_3 \geq \beta_2$ and $\beta_4 \notin [\beta_2,\beta_1]$, the choice for $x^*$ for $R_{i2}$ is $x_{i2}$.
Same for $R_{i3}$. 

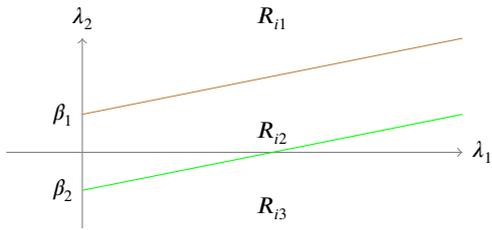
\begin{figure}
\centering
\begin{tikzpicture}
    \draw [thin, gray, ->] (0,-1) -- (0,1.5)      
        node [above, black] {$\lambda_2$};  
    
    \draw [thin, gray, ->] (-1,0) -- (5,0)      
        node [right, black] {$\lambda_1$};    
    
    \draw [draw=brown,] (0,0.5) -- (5,1.5);	
    \draw [draw=green,] (0,-0.5) -- (5,0.5);	

    \node [left] at (0,0.5) {$\beta_1$};      
    \node [left] at (0,-0.5) {$\beta_2$};      

   \node [above] at (2.5,1.5) {$R_{i1}$};      
   \node [above] at (2.5,0) {$R_{i2}$};      
   \node [above] at (2.5,-1) {$R_{i3}$};      

\end{tikzpicture}
\caption{Case 4: Cut Lines for $(\lambda_1 \geq 0, \lambda_2)$ Half-Plane}
\end{figure}
\end{case}
Collecting results for all cases gives us the result for the proposition.
\end{proof}
\endgroup
\end{document}